\date{}
\documentclass[11pt]{article}
\usepackage{hyperref}
\usepackage{tikz}
\usetikzlibrary{snakes}
\usetikzlibrary{decorations.text,calc,arrows.meta}
\usepackage[vcentermath]{youngtab}
\usepackage{tkz-euclide}
\usepackage{pgfplots}
\usetikzlibrary{decorations.shapes}
\usetikzlibrary{positioning}
\usetikzlibrary{patterns}
\usetikzlibrary{shapes.geometric}
\usetikzlibrary{decorations.pathmorphing}
\usetikzlibrary{arrows.meta}
\tikzset{snake it/.style={decorate, decoration=snake}}
\usetikzlibrary{arrows,shapes,positioning}
\usetikzlibrary{decorations.markings}
\tikzstyle arrowstyle=[scale=1]
\tikzstyle directed=[postaction={decorate,decoration={markings,mark=at position .65 with {\arrow[arrowstyle]{stealth}}}}]
\tikzstyle reverse directed=[postaction={decorate,decoration={markings,mark=at position .65 with {\arrowreversed[arrowstyle]{stealth};}}}]

\usepackage{amsmath,amsthm,amsfonts,amssymb,amsopn,amscd} 
\usepackage{color}

\def\qed{{\unskip\nobreak\hfil\penalty50
\hskip2em\hbox{}\nobreak\hfil$\square$
\parfillskip=0pt \finalhyphendemerits=0\par}\medskip}
\def\proof{\trivlist \item[\hskip \labelsep{\bf Proof.\ }]}

\def\tilde{\widetilde}

\def\l{\lambda}

\def\om{\omega}
\def\Om{\Omega}

\def\proof{\trivlist \item[\hskip \labelsep{\bf Proof.\ }]}

\def\eproof{\null\hfill\qed\endtrivlist\noindent}

\newcommand{\bthm}{\begin{theorem}}
\newcommand{\ethm}{\end{theorem}}

\newcommand{\bprop}{\begin{proposition}}
\newcommand{\eprop}{\end{proposition}}
\newcommand{\bcor}{\begin{corollary}}
\newcommand{\ecor}{\end{corollary}}
\newcommand{\blem}{\begin{lemma}}
\newcommand{\elem}{\end{lemma}}

\def\RR{{\mathbb R}}

\def\C{{\cal C}}
\def\D{{\cal D}}

\def\M{{\cal M}}
\def\N{{\cal N}}
\def\R{{\cal R}}

\def\H{{\cal H}}

\def\S{{\cal S}}

\def\f{{\varphi}}

\def\l{{\lambda}}

\def\PSL{{{\rm PSL}(2,\mathbb R)}}

\def\S2{S^{1(2)}}

\def\RR{\mathbb R}
\newtheorem{theorem}{Theorem}[section]
\newtheorem{lemma}[theorem]{Lemma}

\newtheorem{corollary}[theorem]{Corollary}

\newtheorem{proposition}[theorem]{Proposition}

\theoremstyle{definition} 

\theoremstyle{remark} \newtheorem{remark}[theorem]{Remark}

\newcommand{\ben}{\begin{equation}}
\newcommand{\een}{\end{equation}}

\def\PSL{PSU(1,1)}

\def\CC{{\mathbb C}}

\def\SL2{{{\rm SL}(2,\R)}}

\def\PSL2{{{\rm PSL}(2,\Reali)}}

\def\U1{{{\rm V}(1)}}
\def\SU2{{{\rm SV}(2)}}

\def\SU{{{\rm SU}}}

\def\C{{\mathcal C}}
\def\D{{\mathcal D}}

\def\H{{\mathcal H}}

\def\M{{\mathcal M}}
\def\N{{\mathcal N}}

\def\P{{\mathcal P}}

\title{\Huge{A New Proof of the QNEC
}}

\author{ {\sc Stefan Hollands}\\
Institut für Theoretische Physik, Universität Leipzig \\
Brüderstrasse 16, 04103 Leipzig, Germany\\
Max Planck Institute for Mathematics in Sciences (MiS)\\
Inselstra{\ss}e 22, 04103 Leipzig, Germany
\\
{}
\\
{\sc Roberto Longo}\\
Dipartimento di Matematica,
Tor Vergata Universit\`a di Roma\\
Via della Ricerca Scientifica, 1, I-00133 Roma, Italy
}

\date{}
\begin{document}

\maketitle

\begin{abstract}
We give a simplified proof of the quantum null energy condition (QNEC). Our proof is based on an explicit formula for the shape derivative of the relative entropy, with respect to an entangling cut. It allows bypassing the analytic continuation arguments of a previous proof by Ceyhan and Faulkner and can be used e.g., for defining entropy current fluctuations. 
\end{abstract}

\section{Introduction}

The quantum null energy condition (QNEC) states that\footnote{We use units such that $k_B=\hbar=c=1.$} 
\ben
\label{QNEC1}
2\pi \langle T_{\alpha\beta} \rangle k^\alpha k^\beta \ge S_{\rm EE}'', 
\een
where the double prime indicates the second shape variation, of an entangling cut, in the direction $k^\alpha$ within an outgoing, non-expanding, affinely parameterized null-surface tangent and normal to $k^\alpha$. $S_{\rm EE}$ is the entanglement entropy associated with the cut and some state of the quantum field theory (QFT), and $\langle T_{\alpha\beta} \rangle$ is the expected stress energy tensor\footnote{In a curved spacetime, the operator $T_{\alpha\beta}$ has well-known ambiguities, and it is somewhat unclear how these are understood in \eqref{QNEC1}. However, such ambiguities are not present in $T_{\alpha\beta}k^\alpha k^\beta$ if the null surface is a (future) bifurcate Killing horizon, where $C_{\alpha\beta}k^\alpha k^\beta=0$ for any covariant tensor $C_{\alpha\beta}$ made from contractions of the Riemann tensor and its covariant derivatives 
\cite{bobred}.} of the QFT in that state.

The QNEC can be seen as a semi-classical limit of the quantum focusing conjecture \cite{B1}, potentially a fundamental feature of quantum gravity.  
Among other things, it is the basis of the generalized second law for dynamical {\it apparent} black hole horizons \cite[Sec. VII]{bobred}.

In \cite{B2}, a heuristic argument in favor of the QNEC was given.
But it is not straightforward to obtain a mathematically rigorous proof---nor even statement---of the QNEC, because both the precise formulation of second shape variation as well as the notion of entanglement entropy are subtle in QFT. 

Ceyhan and Faulkner proposed \cite{CF18} that ``half-sided modular inclusions'' \cite{Bor1, Bor2, Wies1, AZ05,Florig} associated with the entangling cut are a natural framework for the QNEC. By this one means an inclusion $\N \subset \M$ of von Neumann algebras together with a pure state $\Omega$
such that the dynamics associated with the reduction of that state to $\M$, called the ``modular flow'', cannot leave $\N$ for positive flow times, see sections \ref{modthbas}, \ref{hsm} for precise definitions.

A manifestation of this structure is given by the set-up illustrated in figure \ref{fig:1}: $\M$ is the algebra of QFT observables associated with the entire exterior region of Schwarzschild spacetime, $\N$ the subalgebra associated with an entangling cut $C_a$ of the future horizon, and $\Omega$ is the Hartle-Hawking state. The modular flow corresponds to the time-translation symmetry of the Schwarzschild spacetime.

\begin{figure}[h!]
\label{fig:1}
\centering
\begin{tikzpicture}[scale=.5]
\draw[snake] (-2,5) -- (6,5);
\draw[snake] (6,-3) -- (-2,-3);
\filldraw[color=black, fill=darkgray, thin] (6,-3) -- (10,1) -- (6,5) -- (2,1) -- (6,-3);
\filldraw[color=black, fill=lightgray, thin] (6,-3) -- (8,-1) -- (4,3) -- (2,1) -- (6,-3);
\filldraw[color=black, fill=white, thin] (-2,-3) -- (-6,1) -- (-2,5) -- (2,1) -- (-2,-3);
\draw[ thick, dashed] (-2,-3) -- (6,5);
\draw[thick, dashed] (6,-3) -- (-2,5);
\draw (6,5) node[anchor=south]{$i^+$};
\draw (3.2,2.5) node[anchor=east]{$H^+$};
\draw (4.2,3.5) node[anchor=east]{$C_a$};
\draw (0.3,-.3) node[anchor=east]{$\bar H^-$};
\draw (8.4,3) node[anchor=west]{$I^+$};
\draw (-4.4,-1) node[anchor=east]{$\bar I^-$};
\draw[->,thick] (5.7,-2.3) .. controls (2.5,1)  .. (5.7,4.3);
\draw[->,thick] (6.3,-2.3) .. controls (9.5,1)  .. (6.3,4.3);
\end{tikzpicture}
\caption{$\N$ corresponds to the dark gray region, $\M$ to the light- and dark gray regions combined. The arrows indicate the orbits of the positively directed modular flow of the Hartle-Hawking state with respect to $\M$.}
\end{figure}
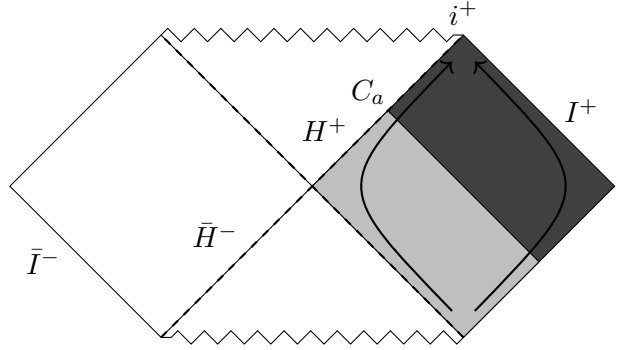

The structural theorem \cite{Wies1,AZ05,Florig} about half-sided modular inclusions states that there always exists a {\it positive} self-adjoint generator $P$ of ``translations'' whose unitary Heisenberg evolution $U(a)=e^{iaP}$ obeys the relations of an affine group $A(1)$ with the modular flow of $\Omega$. In the above example of Schwarzschild spacetime, $U(a)$ indeed implements affine translations of the dark gray wedge, sliding it along $H^+$. 
It is noteworthy that such translations are {\it not} isometries of Schwarzschild spacetime, so the existence of $P$ is non-trivial. Heuristically, $P$ is an integral over the null-components of the stress tensor\footnote{Informally, $P=\int_{H^+ \cup \bar H^- \cup I^+ \cup \bar I^-} T_{\alpha\beta}k^\alpha dS^\beta$, see figure \ref{fig:1}.}, though one should stress that the notion of half-sided modular inclusion does not require this object a priori. 

A half-sided modular inclusion defines a whole one-parameter family of nested {\color{blue} (decreasing in $a$)} algebras by $\M(a) = U(a) \M U(a)^*$. For each $a$ and each state $\Phi$ in the Hilbert space, one may then define Araki's relative entropy \cite{Ar76,Ar77}
\ben
S(a) := S(\Phi |\! | \Omega)_{\M(a)}
\een
with respect to $\M(a)$. Ceyhan and Faulkner argued \cite{CF18} that a mathematically rigorous reformulation of the QNEC \eqref{QNEC1} should be that $S$ is a
convex function, meaning that 
\ben
\label{par2pos}
\partial^2 S(a) \ge 0 
\een
if it is twice differentiable. 
Their formulation nicely avoids the technical problems with \eqref{QNEC1} in two ways: arbitrary shape variations are included because the inequality holds for a half-sided modular inclusion associated with an {\it arbitrary} entangling cut.
Furthermore, the relative entropy is better defined mathematically than the entanglement entropy, and formally yields the terms on the left and right side of \eqref{QNEC1} at the same time. 

More concretely, the connection between \eqref{par2pos} and \eqref{QNEC1} can be seen from the heuristic formulas
\ben
\omega_a = e^{-2\pi K_a}, \quad 
K_a = \int\limits_{V>a} (V-a) T_{\alpha\beta}k^\alpha dS^\beta + \dots , 
\een
for the reduced density matrix, $\omega_a$, of $\Omega$ to the dark gray region in figure \ref{fig:1} that is defined by the cut $C_a$. Here, $V$ is an affine parameter on $H^+$ such that $k^\alpha\nabla_\alpha V = 1$, $V=a$ on $C_a$
and such that $V=0$ on the bifurcation surface.
The dots indicate additional contributions, for example,
from $I^+$ and $i^+$, as well as a formally infinite constant, all of which arguably do not depend on $a$. These disappear when the two derivatives are applied in \eqref{par2pos} using the heuristic formula $S(a) = {\rm Tr}(\varphi_a \log \varphi_a - \varphi_a \log \omega_a)$ for the relative entropy in terms of the reduced density matrix $\varphi_a$ of $\Phi$. A short calculation thereby gives a formula equivalent to \eqref{QNEC1} for an appropriate notion of second shape variation.

A proof of the QNEC, in their formulation, was also given by \cite{CF18}, based on the ``ant-formula'' suggested by a parable in which a massless
‘ant’ tries to estimate the null-energy from observations along its trajectory \cite{W1}. The ant-formula is 
\ben
\label{antleftintro}
-\partial S(a) = 2\pi \inf_{u'} (u'\Phi, P u'\Phi)
\een
in the formulation by \cite{CF18}, where $u'$ runs over unitary operators from the commutant $\M(a)'$ of $\M(a)$. The variational nature of this formula immediately gives\footnote{Note that monotonicity of the relative entropy only gives $S(b) \le S(a)$ when $b \ge a$.} $\partial S(b) \ge \partial S(a)$ when $b \ge a$, because $\M(a)' \subset \M(b)'$. We thereby obtain the QNEC even in cases when only the first, but not second, derivative of $S(a)$ exists.

While \cite{CF18} could show without too much difficulty that the right side of \eqref{antleftintro} cannot be smaller than 
$-\partial S(a)$, the construction of a sequence of unitaries $u'_s$ saturating  \eqref{antleftintro} turned out to be the central technical problem in their proof, requiring complicated techniques of analytic continuation in two variables and advanced methods from complex analysis. 

In the present paper, we will give a substantially simplified argument in this crucial step, avoiding any such analytic continuations. Like the construction by \cite{CF18}, we shall choose the minimizing sequence $u'_s$ to be the Connes-cocycle \cite{C73} between $\Omega, \Phi$ sending the parameter $s \to \infty$.
However, our arguments why this saturates the ant-formula are different from \cite{CF18} and rely on an explicit formula for the derivative of the relative entropy which we derive in section \ref{paSform}:
\ben
\label{prop2s1intro}
\partial S(a) = i\left(\Phi, [P, \log \Delta_a'] \Phi \right),
\een
where $\Delta'_a$ is a relative modular operator \cite{Ar76,Ar77} between $\Phi,\Omega$ with respect to $\M(a)'$. From \eqref{prop2s1intro} to the proof of the QNEC is a relatively short step, described in section \ref{qnecproof}.

Equation \eqref{prop2s1intro} might be of independent interest because it displays the 
relative entropy flux as an expectation value of a suitably defined entropy-current operator. In particular, 
one may consider the variance (fluctuations) of the relative entropy flux, see section \ref{outlook} for further discussion. 

We finally note that in this paper we derive the QNEC for a natural dense set of vector states. The set of vectors considered by \cite{CF18} is possibly larger but not complete. Proving the QNEC for all vector states remains an open problem.

\section{Relative modular operators and entropy}
\label{modthbas} 

Let $\M$ be a von Neumann algebra on a Hilbert space $\H$, and $\Om,\Phi\in\H$ cyclic and separating vectors for $\M$. The {\it  relative Tomita's operators} \cite{Ar76,Ar77} on $\H$ are given by the closures of
\[
S_{\Om,\Phi} \equiv S_{\Om,\Phi;\M} : x\Phi \mapsto x^*\Om\, , \quad x\in \M\, ,
\]
\ben\label{SF}
S'_{\Om,\Phi} \equiv S_{\Om,\Phi;\M'}: x'\Phi \mapsto {x'}^*\Om\, , \quad x'\in \M'\, .
\een
By considering the polar decompositions, we get the {\it relative modular operators and conjugations}:
\ben
\label{poldec}
S_{\Om,\Phi} = J_{\Om,\Phi}\Delta^{1/2}_{\Om,\Phi}\, ,\qquad S'_{\Om,\Phi} = J'_{\Om,\Phi}\Delta'^{1/2}_{\Om,\Phi}\, .
\een
Recall the formulas
\ben\label{JJ3}
J'_{\Om,\Phi} = J_{\Phi,\Om} = J^*_{\Om,\Phi}  \, , \qquad  {\Delta}'_{\Om,\Phi} = {\Delta}^{- 1}_{\Phi,\Om} = J_{\Om,\Phi}\Delta_{\Om,\Phi} J^*_{\Om,\Phi} \, .
\een
From the last equality in \eqref{JJ3} we have
\[
J_{\Om,\Phi}  f({\Delta}_{\Om,\Phi}) J^*_{\Om,\Phi}  = \bar f(\Delta^{-1}_{\Phi,\Om})
\]
for every complex Borel function $f$ on $(0, \infty)$, with $\bar f$ the complex conjugate of $f$. 
The modular operators have the covariance properties:
\ben
\label{covariance}
{\Delta}_{v\Om,u\Phi} = v{\Delta}_{\Om,\Phi}v^*, \quad
{\Delta}_{v'\Om,u'\Phi} =  u'{\Delta}_{\Om,\Phi}u'^*,
\een
for any unitaries $u,v \in \M, u', v' \in \M'$.
We also use the modular flow, which are the 1-parameter groups of 
automorphisms of $\M$ respectively $\M'$ given by
\ben
\label{mflow1}
\sigma^\Phi_t(m) = \Delta_\Phi^{it} m \Delta_\Phi^{-it}, \quad \sigma'^\Phi_t (m') = \Delta_\Phi^{-it} m' \Delta_\Phi^{it}
\een
respectively. Here, and in the rest of the paper, we define $\Delta_\Phi := \Delta_{\Phi,\Phi}$. The modular flows may also be expressed with the help of the relative modular operators:
\ben
\label{mflow2}
 \Delta_{\Phi,\Omega}^{it} m \Delta_{\Phi,\Omega}^{-it} = \Delta^{it}_\Phi m \Delta^{-it}_\Phi ,
 \quad 
  \Delta_{\Phi,\Omega}^{it} m' \Delta_{\Phi,\Omega}^{-it} =  \Delta^{it}_\Om m' \Delta_{\Omega}^{-it}.
\een
With $\f = (\Phi, \cdot \Phi)$, $\om = (\Om, \cdot \Om)$ the states on $\M$ associated with $\Phi,\Om$, the following formulas for the {\it Connes-cocycles} \cite{C73} hold:
\ben\label{uDD}
u_s =(D\om : D\f)_s = \Delta^{is}_{\Om,\Phi}\Delta^{-is}_{\Phi} =
\Delta^{is}_{\Om}\Delta^{-is}_{\Phi,\Om}
\, ; 
\een
$u_s$ respectively $u_s'$ are unitary operators from $\M$ respectively $\M'$ for all $s \in \RR$.
Since $(D\om :D \f)_s = (D\f : D\om)^*_s$, we have by \eqref{JJ3}
\ben\label{uu}
u_s = \Delta^{is}_{\Om,\Phi}\Delta^{-is}_{\Phi} = \big(\Delta^{is}_{\Phi,\Om}\Delta^{-is}_{\Om}\big)^*
= \Delta^{is}_{\Om}\Delta^{-is}_{\Phi,\Om} = \Delta^{is}_{\Om}\Delta'^{is}_{\Om,\Phi}\, .
\een
Set $u'_s  =(D\om' : D\f')_s$, with $\f' = (\Phi, \cdot \Phi)$, $\om' = (\Om, \cdot \Om)$ on $\M'$. 
Then 
\ben\label{uu2}
u_{-s} u'_{s} = \Delta^{-is}_{\Om}\Delta^{is}_{\Phi} \, .
\een
One may also define unitary cocycles associated with the modular conjugations \cite[App. C]{AM82}, \cite[App. A]{CF18}:
\ben
v_{\Omega,\Phi} = J'_{\Omega,\Phi} J'_\Phi = J'_{\Omega} J'_{\Omega,\Phi} \in \M,
\een
and similarly for $\M'$.

If $\Phi$ is not cyclic or not separating, then appropriate modifications of the above formulas 
involving the so-called support projections $s(\Phi) \in \M, s'(\Phi) \in \M'$ apply, where for instance
$s(\Phi)$ is the orthogonal projection onto the closure of the subspace $\M'\Phi$. 
For details on such relations see \cite{Ar77}, \cite[App. C]{AM82}, \cite[App. A]{CF18}. 

For a cyclic and separating vector $\Omega$, Araki's {\it relative entropy} is defined by \cite{Ar76,Ar77}
\ben
\label{Sdef}
S(\Phi|\!| \Om)_\M =  - (\Phi , \log \Delta_{\Om,\Phi}\Phi). 
\een
The first covariance property \eqref{covariance} implies that 
\ben
\label{Sdefcov}
S(u'\Phi|\!| v'\Om)_\M = S(\Phi|\!| \Om)_\M, \quad 
S(u\Phi|\!| v\Om)_{\M'} = S(\Phi|\!| \Om)_{\M'},
\een
for any isometries $u,v \in \M, u',v' \in \M'$. This implies that $S(\Phi|\!| \Om)_\M \equiv S(\varphi|\!| \omega)_\M$ i.e., the relative entropy only depends on the 
functionals $\f = (\Phi, \cdot \Phi)$, $\om = (\Om, \cdot \Om)$, the states on $\M$ associated with $\Phi,\Om$, and similarly for the commutant (second formula). 

\blem\label{mD}
Let $\M$ be a von Neumann algebra on the Hilbert space $\H$ and $\Phi, \Om\in\H$ cyclic and separating vectors. Then
there exists $m_\l\in \M$ such that
\ben\label{f1}
(\Delta_{\Phi,\Om} + \l)^{-1}\Phi = m_\l\Om
\een
for any $\l >0$. 
In addition, 
the map $\l \mapsto m_\l$ is strongly continuous and we have
\ben\label{f2}
\int_0^\infty m_\l m_\l^* d\l = 1
\een
(integral in the weak topology). 
\elem
\proof
We have
\[
(\Delta_{\Phi,\Om} + \l)^{-1} = (\Delta^{1/2}_{\Phi,\Om}  + \l \Delta^{-1/2}_{\Phi,\Om} )^{-1} \Delta^{-1/2}_{\Phi,\Om}
= (\Delta^{1/2}_{\Phi,\Om}  + \l \Delta^{-1/2}_{\Phi,\Om} )^{-1} \Delta'^{1/2}_{\Om,\Phi}\, .
\]
From the formula
\[
 \frac{1}{e^{p/2} + e^{-p/2}} = 
\int_{-\infty}^\infty\frac{e^{itp}}{e^{\pi t} + e^{-\pi t}}dt 
\]
we then get
\ben
\label{f3}
(\Delta_{\Phi,\Om} + \l)^{-1}\Phi
= \,\frac{1}{2\sqrt{\lambda}} \int_\RR \frac{\lambda^{-it}}{\cosh(\pi t)} \Delta_{\Phi,\Omega}^{it} \, dt \, {\Delta'^{1/2}_{\Om,\Phi} }\Phi
= \, \frac{1}{2\sqrt{\lambda}} \int_\RR \frac{\lambda^{-it}}{\cosh(\pi t)} \Delta_{\Phi,\Omega}^{it}  \, dt \, J_{\Omega, \Phi} \Omega \, .
\een 
Using the properties of the relative modular flow and the definition of the Connes-cocycle, we have
\[
\Delta_{\Phi,\Omega}^{it} J_{\Omega, \Phi} \Omega 
= \Delta_{\Phi,\Omega}^{it}  J_{\Phi,\Omega}' J_\Omega' \Omega
= \Delta_{\Phi,\Omega}^{it} w \Omega
= \Delta_{\Phi,\Omega}^{it}  w \Delta_{\Phi,\Omega}^{-it} \Delta_{\Phi,\Omega}^{it} \Delta_{\Omega}^{-it}\Omega
= \sigma^{\f}_{t}(w) (D\varphi:D\omega)_{t} \Omega\, , 
\]
where $w = J_{\Phi,\Omega}' J_\Omega'\in \M$ is a unitary and $\om$, $\f$ are the states on $\M$ associated with $\Om$, $\Phi$. 

Therefore formula \eqref{f1} holds with
\ben\label{M1}
m_\l = \frac{1}{2\sqrt{\lambda}} \int_\RR \frac{\lambda^{-it}}{\cosh(\pi t)} v_t\, dt 
\een
where $v_t = \sigma^{\f}_{t}(w) (D\varphi:D\omega)_{t}$. 

The strong continuity of $\l \mapsto m_\l$ follows from \eqref{M1} and the Lebesgue dominated convergence theorem.

We now prove the relation \eqref{f2}. 
The map $t\mapsto v_t^{*} \xi/\cosh(\pi t)$ is an $\H$-valued $L^2$-function on $\RR$ for all $\xi \in \H$. 
Therefore, since the Fourier transform is an isometry of $L^2(\RR, \H)$, also $u \mapsto e^{u/2} m^{*}_{e^u} \xi$ is an $\H$-valued $L^2$-function on $\RR$.
Using that the Fourier transform is an isometry of $L^2(\RR, \H)$ and making a change of integration variable from $\lambda$ to $u=\log \lambda$, we therefore have by \eqref{f3}, for any pair $\xi,\eta \in \H$:
\ben
\label{f4}
\int_{(0,\infty)} \left( \xi, m_\lambda m^{*}_\lambda \eta \right) d\lambda
= \frac\pi2 \int_{\RR}  \frac{\left( \xi, v_t v^{*}_t\eta \right) }{\cosh^2(\pi t)}dt
= \frac\pi2 \int_{\RR}  \frac{\left( \xi,  \eta \right) }{\cosh^2(\pi t)}dt
= (\xi,  \eta)
\een
because $v_t$ is unitary.  
\eproof
\begin{remark}
By the same argument, Lemma \ref{mD} shows that 
$(\Delta_{\Phi,\Om} + \l)^{-1}\M'\Phi \subset \M\Omega$. The case $\Phi = \Omega$ is Tomita's lemma, used in most proofs of the Tomita-Takesaki main theorem.
\end{remark}
Let now $T_n, T$ be  (anti)-linear, densely defined, closed operators on the Hilbert space $\H$. Denote by $P_T$ the orthogonal projection on $\H\oplus\H$ onto the graph of $T$. We write $T_n \to_gT$  if $P_{T_n} \to P_T$ strongly. 

As is known \cite{Stone}, we have
\[
P_T = 
\begin{pmatrix}(1 + T^*T)^{-1} & T^* (1 + TT^*)^{-1}  \\
T (1 + T^*T)^{-1}  & (1 + TT^*)^{-1} 
\end{pmatrix}
\]
therefore $T_n \to_gT$ implies that $(1 + T_n^*T_n)^{-1} \to (1 + T^*T)^{-1} $ strongly. 
\bprop\label{cont1}
Let $\M$ be a von Neumann algebra on the Hilbert space $\H$ and $\Phi_n, \Om_n, \Phi, \Om \in \H$ cyclic and separating vectors 
with $\Phi_n \to \Phi$ and $\Om_n \to \Om$ in norm. Then $(\Delta_{\Phi_n, \Om_n} + \l)^{-1} \to (\Delta_{\Phi, \Om} + \l)^{-1}$
and $J_{\Phi_n, \Om_n}  \to J_{\Phi, \Om} $ strongly for all $\l >0$. 
\eprop
\proof
It is easily seen that $S_{\Phi_n, \Om_n} \to_g S_{\Phi, \Om}$ thus the proposition follows by the above considerations by arguments similar to the ones in \cite[Sect. 6]{HW} for the case $\Phi_n = \Om_n$.
\eproof

\section{Half-sided modular inclusions}
\label{hsm}

Let $\N \subset \M$ be an inclusion of von Neumann algebras on a Hilbert space $\H$ such that the following conditions are satisfied:
\begin{enumerate}
\item There exists a unit vector $\Omega \in \H$ which is cyclic and separating for both $\M$ and $\N$. 
\item Denote by $\Delta_\Omega$ the modular operator for $\M$. Then $\Delta^{-it}_\Omega \N \Delta^{it}_\Omega \subset \N$ for all $t \ge 0$. 
\end{enumerate}
Then $(\N \subset \M, \Omega)$ is called a {\it half-sided modular inclusion} with respect to $\Omega$. Given a half-sided modular inclusion, one can define the family of unitary operators 
\ben\label{Wiesbrock}
U(1-e^{-2\pi t}) = \Delta^{it}_{\Omega; \N }\Delta^{-it}_{\Omega;\M}. 
\een
 Wiesbrock's theorem \cite{Wies1,AZ05,Florig} is the statement that:
 
 \begin{theorem}
 \label{thm1}
 Given a half-sided modular inclusion $(\N\subset\M, \Omega)$ there is a family of unitary operators $U(a), a \in \RR$, given by \eqref{Wiesbrock} for $a \le 1$, realizing  
 the situation described by Borchers' theorem \cite{Bor1,Bor2}, namely one has:
 \begin{enumerate}
 \item $\{ U(a) \mid a \in \RR\}$ is a strongly continuous 1-parameter group of unitaries with self-adjoint generator $P$, $U(a) = e^{iaP}$.
$P$ is positive, meaning ${\rm spec} P  \subset [0,\infty) $.
 \item $U(a) \Omega = \Omega$ for all $a \in \RR$.
 \item $\M(a) := U(a) \M U(a)^* \subset \M$ for $a \ge 0$ and $\M(a) \subset \M$ is a half-sided modular inclusion relative to $\Omega$. In particular, $\M(b) \subset \M(a)$ for $b \ge a$.
 \item $\M(1) = \N$.
 \item $\Omega$ is cyclic and separating for each $\M(a)$ and therefore for each $\M(a)'$.
 \item $\Delta^{-it}_{\Omega} P \Delta^{it}_{\Omega} = e^{2\pi t} P$ for all $t \in \RR$, on the domain $\D(P)$ of $P$ given by Stone's theorem (in particular, $\D(P)$ is invariant under $\Delta^{-it}_{\Omega}$), or equivalently $\Delta^{-it}_{\Omega} U(a) \Delta^{it}_{\Omega} = U(e^{2\pi t} a)$ for all $t,a \in \RR$.
 \item $U(-a) \M' U(-a)^* \subset \M'$ for $a \ge 0$ and $J_{\Omega} U(a) J_{\Omega} = U(-a)$ for all $a \in \RR$.
 \end{enumerate}
 \end{theorem}
By item 3) and the monotonicity of the relative entropy \cite{U77}, the function $a \mapsto S(\Phi|\!| \Om)_{\M(a)}$ from $\RR \to [0,\infty]$ is monotonically decreasing. It is also clear from \eqref{da1} that 
\ben
\label{Sadef}
S(a):=S(\Phi|\!| \Om)_{\M(a)} = S(U(a)^* \Phi |\!| U(a)^* \Om)_\M =  S(\Phi_a |\!| \Om)_\M, \quad \Phi_a := U(a)^* \Phi. 
\een
Since $a \mapsto \Phi_a$ is strongly continuous, it follows that $a \mapsto (\Phi_a, \! \cdot  \ \Phi_a)$ is weak-$^*$ continuous. By the lower semi-continuity 
of the relative entropy \cite[Thm. 3.7]{Ar77}, therefore 
$a \mapsto S(a)$ and $a \mapsto \bar S(a)$ are lower semi-continuous functions from $\RR \to [0,\infty]$.
The structure of half-sided modular inclusions imply a stronger form of continuity, as follows. Suppose
$a,b\in\RR$ with $S(a) <\infty$, ${\bar S}(b) <\infty$. Then by \cite[Lem. 1]{CF18} or \cite[Prop. 3.2]{HL25}, we have the {\it sum rule}
\ben\label{dS1}
S(a) - S(b)   =   {\bar S}(a) - {\bar S}(b) + 2\pi(b-a) (\Phi, P\Phi) \, ,
\een
for every vector state $\Phi\in \D(P)$. For any increasing function $f:I=[a,b] \to \RR$, we have a Lebesgue decomposition 
$f=f_A+f_C+f_J$, unique up to additive constants, into increasing functions such that $f_A$ is absolutely continuous, $f_C$ is continuous and has $\partial f_C=0$ almost everywhere on $I$, and $f_J$ is a jump function, see e.g., \cite[Ch. 3]{Stein}. We apply this to the increasing functions $f=-S$ and $f=\bar S$. Since $\partial(-S+\bar S)=2\pi (\Phi, P\Phi)$ by \eqref{dS1}, the contributions 
from jump function pieces cannot be present in the Lebesgue decompositions of either $-S,\bar S$. 
In particular, $S,\bar S$ are continuous and their derivatives $\partial S, \partial \bar S$ exist almost everywhere on $I$.

\section{Formula for $\partial S(a)$}
\label{paSform}

For a positive self-adjoint operator $A$ 
with $\ker A = \{0\}$, we have the integral formula 
\ben\label{L1}
\log A = \int_{(0,\infty)}  \left[(1+\lambda)^{-1} -(A+\lambda)^{-1} \right] \, d\lambda \, ,
\een
meaning that
\ben
(\Phi, \log A \, \Phi) = \int_{(0,\infty)}  \left[(1+\lambda)^{-1}\|\Phi\|^2 -(\Phi, (A+\lambda)^{-1}\Phi) \right] \, d\lambda 
\een
for all $\Phi\in \H$ such that $(\Phi, \log A \Phi) $ is well-defined\footnote{For a self-adjoint operator $B$
we say that $(\xi, B\xi)$ is well-defined if either $(\xi,B_+\xi)<\infty$ or $(\xi, B_-\xi)>-\infty$, where $B_\pm$ denote the positive and negative part of the operator.}. 

This formula will allow us to reduce the considerations about the logarithm of the modular operator to its resolvent $(A+\lambda)^{-1}$. In particular, from 
\eqref{Sdef}, \eqref{JJ3}, we have 
\ben
\label{Sresloldef}
S(\Phi|\!| \Om)_{\M} = \int_{(0,\infty)} \left( \Phi, \left[(1+\lambda)^{-1} -(\Delta'_{\Phi,\Omega}+\lambda)^{-1} \right] \Phi \right) \, d\lambda.
\een
Thereby, with our previous notation \eqref{Sadef} for $S(a)$, we also have
\ben
S(a)-S(b) = \int_{(0,\infty)} \left( \Phi, \left[(\Delta'_{b}+\lambda)^{-1} -(\Delta'_{a}+\lambda)^{-1} \right] \Phi \right) \, d\lambda,
\een
where from now on, we will use shorthands such as 
\ben\label{Da}
\Delta_a := \Delta_{\Phi,\Omega;\M(a)}, \qquad \Delta_a' := \Delta_{\Phi,\Omega;\M(a)'}.
\een
Note that we have
\ben\label{da1}
\Delta_a := \Delta_{\Phi,\Omega;\M(a)} = U(a) \Delta_{U(a)^*\Phi,\Omega;\M} U(a)^* = U(a) \Delta_{\Phi_a ,\Omega;\M} U(a)^* \, ,
\een
where $\M = \M(a)$ and $\Phi_a = U(a)^* \Phi$; and similarly
\ben\label{da2}
\Delta_a' := \Delta_{\Phi,\Omega;\M(a)'} = U(a) \Delta'_{\Phi_a ,\Omega;\M} U(a)^* \, .
\een

\blem
\label{lem1}
Let $\Phi \in \H, \Phi_a := U(a)^* \Phi, \lambda > 0, a \in \RR$. Then 
\ben
\left( \Phi, (\Delta'_{a}+\lambda)^{-1} \Phi \right) = \left(\Phi_a, (\Delta'_{\Phi_a,\Omega}+\lambda)^{-1} \Phi_a \right).
\een
\elem
\begin{proof}
From \eqref{da2}, we have $(\Delta'_{a}+\lambda)^{-1} = U(a) (\Delta'_{\Phi_a ,\Omega;\M} + \l)^{-1} U(a)^* $, and this immediately gives the lemma.  
\end{proof}
\blem
\label{lem2}
The map $a \mapsto (\Delta'_{\Phi_a,\Omega}+\lambda)^{-1}$ is strongly continuous for any $\lambda>0$. 
\elem
\begin{proof}
It follows by the continuity $\Psi \mapsto \Delta_{\Psi, \Om}$ in the strong resolvent sense given by proposition  \ref{cont1}. 
\end{proof}
For $\lambda>0$, we now investigate the limit
\ben
\label{reslim}
\begin{split}
& \, \liminf_{b\to a} \, (b-a)^{-1} \left( \Phi, \left[(\Delta'_{b}+\lambda)^{-1} -(\Delta'_{a}+\lambda)^{-1} \right] \Phi \right) \\
=& \, \liminf_{b\to a} \, (b-a)^{-1} \left[ \left( \Phi_b, (\Delta'_{\Phi_b,\Omega}+\lambda)^{-1} \Phi_b \right)  - \left(\Phi_a, (\Delta'_{\Phi_a,\Omega}+\lambda)^{-1} \Phi_a \right) \right] ,
\end{split}
\een
where lemma \ref{lem1} is used to prove the equality. We assume throughout that $\Phi \in \D(P)$. By Stone's theorem $\D(P) = 
\{ \Psi \in \H \mid s-\lim_{a \to 0} [U(a)\Psi-\Psi]/a \,\text{exists} \}$, and for $\Psi \in \D(P)$, we have in fact $iP\Psi = s-\lim_{a \to 0} [U(a)\Psi-\Psi]/a$ or equivalently, 
$-iP\Psi_a = s-\lim_{b \to a} (\Psi_b-\Psi_a)/(b-a)$ in our notation $\Psi_a = U(a)^*\Psi$. Therefore, since $a \mapsto (\Delta'_{\Phi_a,\Omega}+\lambda)^{-1}$ is strongly continuous
and uniformly bounded by lemma \ref{lem2}, it is easily seen from the second line in \eqref{reslim} that
\ben
\label{reslim1}
\begin{split}
& \, \liminf_{b\to a} \, (b-a)^{-1} \left( \Phi, \left[(\Delta'_{b}+\lambda)^{-1} -(\Delta'_{a}+\lambda)^{-1} \right] \Phi \right) \\
=& \, i\left( P\Phi_a, R_a(\lambda)  \Phi_a \right)  - i\left( R_a(\lambda) \Phi_a, P\Phi_a \right)  + \, \liminf_{b\to a} \, (b-a)^{-1} \left( \Phi_a, \left[R_b(\lambda) - R_a(\lambda) \right] \Phi_a \right),
\end{split}
\een
where we use 
\ben
R_a(\lambda):= (\Delta'_{\Phi_a,\Omega}+\lambda)^{-1} 
\een
for the resolvent. Relation \eqref{reslim1} allows us to prove the following lemma.

\blem
\label{lem3}
Let $\Phi \in \D(P) \cap \D(\log \Delta'_a)$. Then 
\ben
\liminf_{b\to a+} \frac{S(a)-S(b)}{b-a} \ge -i\left(\Phi, [P, \log \Delta_a'] \Phi \right).
\een
\elem
\begin{remark}
\label{rem3.4}
By the same proof, if $\Phi \in \D(P) \cap \D(\log \Delta_a)$, then an analogous formula holds also for $\partial \bar S(a)$:  
\ben
\label{lem3eq}
\liminf_{b\to a-} \frac{\bar S(b)-\bar S(a)}{b-a} \ge i\left(\Phi, [P, \log \Delta_a] \Phi \right).
\een
\end{remark}
\begin{proof}
By lemma \ref{lem1}, \eqref{reslim1}, 
\eqref{Sresloldef}, and $\Phi \in \D(\log \Delta'_a)$, we have
\ben
\label{lemaux1}
\liminf_{b\to a+} \frac{S(a)-S(b)}{b-a} \ge -i\left(\Phi, [P, \log \Delta_a'] \Phi \right) + \, \int_{(0,\infty)}
\liminf_{b\to a+} \, \frac{\left( \Phi_a, \left[R_b(\lambda) - R_a(\lambda) \right] \Phi_a \right)}{b-a} d\lambda. 
\een
We should therefore show that the integral in \eqref{lemaux1} is non-negative. Using the following elementary algebraic property of the resolvent, 
\ben
R_b(\lambda) - R_a(\lambda) = R_b(\lambda) \left[ \Delta'_{\Phi_a,\Omega} -  \Delta'_{\Phi_b,\Omega} \right] R_a(\lambda) = R_a(\lambda) \left[ \Delta'_{\Phi_a,\Omega} -  \Delta'_{\Phi_b,\Omega} \right] R_b(\lambda), 
\een
twice, we have for $\lambda >0$,
\ben
\label{lemaux0}
\begin{split}
& \left( \Phi_a, \left[R_b(\lambda) - R_a(\lambda) \right] \Phi_a \right) \\
=& \,
\left( \Phi_a, R_a(\lambda) \left[ \Delta'_{\Phi_a,\Omega} -  \Delta'_{\Phi_b,\Omega} \right] R_a(\lambda) \Phi_a \right) + \\
& \,
\left( \Phi_a, R_a(\lambda) \left[ \Delta'_{\Phi_a,\Omega} -  \Delta'_{\Phi_b,\Omega} \right] R_b(\lambda) 
\left[ \Delta'_{\Phi_a,\Omega} -  \Delta'_{\Phi_b,\Omega} \right] R_a(\lambda)\Phi_a \right)\\
\ge& \, 
\left( \Phi_a, R_a(\lambda) \left[ \Delta'_{\Phi_a,\Omega} -  \Delta'_{\Phi_b,\Omega} \right] R_a(\lambda) \Phi_a \right) .
\end{split}
\een
We therefore have
\ben
\label{lemaux2}
\begin{split}
\liminf_{b\to a+} \frac{S(a)-S(b)}{b-a} \ge& \,  -i\left(\Phi, [P, \log \Delta_a'] \Phi \right) + \\
& \, \int_{(0,\infty)}
\liminf_{b\to a+} \, \frac{\left( \Phi_a, R_a(\lambda) \left[ \Delta'_{\Phi_a,\Omega} -  \Delta'_{\Phi_b,\Omega} \right] R_a(\lambda) \Phi_a \right)}{b-a} d\lambda. 
\end{split}
\een
Let $m'_\l\in \M'$ be given by  Lemma \ref{mD} so that
\[
m'_\l \Om = R_a(\l)\Phi_a\, ,\quad \l >0\, .
\]
We have 
\ben
\label{lemaux5}
\begin{split}
& \, \int_{(0,\infty)}
\liminf_{b\to a+} \, \frac{\left( \Phi_a, R_a(\lambda) \left[ \Delta'_{\Phi_a,\Omega} -  \Delta'_{\Phi_b,\Omega} \right] R_a(\lambda) \Phi_a \right)}{b-a} d\lambda \\
=& \, \int_{(0,\infty)}
\liminf_{b\to a+} \, \frac{\left( m'_\lambda \Omega,  \left[ \Delta'_{\Phi_a,\Omega} -  \Delta'_{\Phi_b,\Omega} \right] m'_\lambda \Omega \right)}{b-a} d\lambda \\
=& \, \int_{(0,\infty)}
\liminf_{b\to a+} \, \frac{\| m^{\prime *}_\lambda \Phi_a\|^2 - \| m^{\prime *}_\lambda \Phi_b\|^2 }{b-a} d\lambda \\
=& \, \int_{(0,\infty)} \Bigg[
-i \left(P \Phi_a, m'_\lambda m^{\prime *}_\lambda \Phi_a\right)
+i \left( m'_\lambda m^{\prime *}_\lambda \Phi_a, P \Phi_a \right) \Bigg] d\lambda.
\end{split}
\een
In the last step, we used $\Phi_a \in \D(P)$ and Stone's theorem. 
The last integral is zero because of \eqref{f4} and because $P$ is Hermitian, so
the proof of the lemma is completed. 
\end{proof}
As we now explain, the same conclusions as in lemma \ref{lem3} can also be obtained under slightly different conditions.
A state $\varphi$ on a von Neumann algebra $\M$ is said to be $c$-comparable to $\omega$ for some $c>0$ if 
\ben
\label{ccomp}
c\omega(m) \le \varphi(m) \le c^{-1}\omega(m) \quad \forall m \in \M_+.
\een
It is well-known that if $\varphi$ is $c$-comparable to $\omega$, and $\Phi,\Omega$ is are any vectors implementing $\varphi,\omega$, 
then $\log \Delta_{\Omega,\Phi} - \log \Delta_\Phi$ is a bounded operator in $\M$ with norm at most $\log c^{-1}$ \cite{Araki4}. It follows that $\Phi \in \D(\log \Delta_{\Phi,\Omega}')$, and since $\M(a) \subset \M$ for $a \ge 0$, and recalling \eqref{da2}, it also follows that $\Phi_a = U(a)^* \Phi \in \D(\log \Delta_{a}')$ for all $a \ge 0$. 

Now assume that also $\Phi \in \D(P)$, and let $\varphi_a = (\Phi_a, \cdot \Phi_a)$  be the state on $\M$ induced by $\Phi_a$. 
Then it follows from the strong continuity of $a\mapsto U(a)$ that $[0,\infty) \owns a\mapsto \varphi_a \in \M_*$ is a continuously differentiable function, 
and it follows that $\varphi_a$ is $c$-comparable to $\omega$ for all $a \ge 0$. 
By  \cite[Prop. 2.5]{W25}, it follows that $a \mapsto S(a)$ is a continuously differentiable function on $[0,\infty)$. Obviously the same 
argument works replacing $\M$ by any $\M(b)$. Combining this with lemma \ref{lem3}, we get:

\blem
\label{lem3.1}
Suppose that $\Phi \in \D(P)$ and that the induced state $\varphi = (\Phi,  \cdot \Phi)$ on $\M(b)$ is $c$-comparable to 
$\omega = (\Omega, \cdot \Omega)$ for some $c>0$. 
Then $[b, \infty) \owns a \mapsto S(a)$ is continuously differentiable, and \eqref{prop2s1} holds for $a \ge b$, i.e. 
$-\partial S(a) \ge -i\left(\Phi, [P, \log \Delta_a'] \Phi \right)$. 
\elem
\begin{remark}
\label{rem3.5}
Assuming instead that $\Phi \in \D(P)$ and that $\varphi'$ is $c$-comparable to $\omega'$ on $\M(b)'$, 
one can likewise argue that \eqref{prop2s2} holds for $a \le b$, i.e. 
$\partial \bar S(a) \ge i\left(\Phi, [P, \log \Delta_a] \Phi \right)$.
\end{remark}

We can also upgrade lemma \ref{lem3} using the sum rule \eqref{dS1}, to get a formula for $\partial S(a)$, if we know that this derivative exists. More precisely, we have:

\bprop
\label{prop2}
Suppose that the derivative $\partial S(a)$ exists and that $\Phi \in \D(P) \cap \D(\log \Delta'_a) \cap \D(\log \Delta_a)$, so, in particular,
$S(a), \bar S(a), (\Phi, P \Phi)<\infty$. Then 
\ben
\label{prop2s1}
\partial S(a) = i\left(\Phi, [P, \log \Delta_a'] \Phi \right),
\een
and the analogous formula also holds for $\partial \bar S(a)$,  
\ben
\label{prop2s2}
\partial \bar S(a) = i\left(\Phi, [P, \log \Delta_a] \Phi \right).
\een
\eprop
\begin{proof}
The sum rule \eqref{dS1}, lemma \ref{lem3} and its analogous version for $\bar S$ in remark \ref{rem3.4} imply 
\ben
\label{prop1aux}
2\pi (\Phi, P\Phi) = -\partial S(a) + \partial \bar S(a) \ge -i\left(\Phi, [P, \log \Delta_a'] \Phi \right) +i\left(\Phi, [P, \log \Delta_a] \Phi \right).
\een
We apply \eqref{uu2} to $\Phi$ and differentiate with respect to $s$ at $s=0$, using Stone's theorem 
and the fact that $\Phi \in \D(\log \Delta'_a) \cap \D(\log \Delta_a)$. 
It follows on the one hand that $\partial_s (u_{-s} u_{s}' \Phi) |_{s=0} = -i(\log \Delta_a) \Phi + i(\log \Delta_a') \Phi$, and on the other hand 
that $\partial_s (u_{-s} u_{s}' \Phi) |_{s=0} = -i(\log \Delta_{\Omega;a})  \Phi$. Therefore, 
$\Phi \in \D(\log \Delta_{\Omega;a})$, and we have
\ben
-\log \Delta_a' \Phi + \log \Delta_a \Phi = \log \Delta_{\Omega;a} \Phi. 
\een
Since $\Phi \in \D(P)$, therefore, 
\ben
-i \left(\Phi, [P, \log \Delta_a'] \Phi \right) +i \left(\Phi, [P, \log \Delta_a] \Phi \right) = i\left( \Phi, [P, \log \Delta_{\Omega;a}] \Phi \right).
\een
The expression on the right may be evaluated noting that by item 6) of Wiesbrock's theorem \ref{thm1} applied to the half-sided modular inclusion $\M(a) \subset \M(a+1)$, we have $(\Phi, \Delta_{\Omega;a}^{-it} P \Delta_{\Omega;a}^{it} \Phi) = e^{2\pi t} (\Phi, P \Phi)$, and noting that $\Phi \in \D(\log \Delta_{\Omega;a})$. Therefore, Stone's theorem may be used when taking a derivative with respect to $t$ at $0$, which gives $2\pi (\Phi, P \Phi)$. Relation \eqref{prop1aux} thereby becomes
\ben
\label{prop1aux2}
2\pi (\Phi, P\Phi) = -\partial S(a) + \partial \bar S(a) \ge -i\left(\Phi, [P, \log \Delta_a'] \Phi \right) +i\left(\Phi, [P, \log \Delta_a] \Phi \right) = 2\pi(\Phi, P \Phi).
\een
The inequality in \eqref{prop1aux2} that stems from the application of the inequalities of lemma \ref{lem3} and the following remark \ref{rem3.4} (where $\liminf$ may be replaced by $\lim$ by our assumption) to $-\partial S(a)$ and $\partial \bar S(a)$ must therefore be an equality. 
The same must therefore be true for the inequalities of lemma \ref{lem3} and the following remark \ref{rem3.4}. This proves the proposition. 
\end{proof}

\section{Proof of the ant-formula}
\label{qnecproof}

Using the results from the previous section, we now show the ant-formula \cite{W1,CF18} expressed by the following theorem. As we have already discussed in the introduction, the QNEC will follow from the ant-formula.

\begin{theorem}
\label{antthm}
Suppose that $\partial S(a)$ exists for some $a \in \RR$, 
$\Phi \in \D(P) \cap \D(\log \Delta_a')$, and  $u_{s}' \Phi \in \D(P)$ for 
$s \in (s_0, \infty)$ for some $s_0$. Then 
\ben
\label{antleft}
-\partial S(a) = 2\pi \inf_{u' \in \M(a)'} (u'\Phi, P u'\Phi).
\een
Here, $u_s' = (D\omega' : D\varphi')_s$ are the Connes-cocycles associated with $\M(a)'$ and $\Phi,\Omega$.
The infimum in \eqref{antleft} is over isometries $u' \in \M(a)'$, and $S(a) = S(\Phi|\!| \Om)_{\M(a)}$.
\end{theorem}
\begin{remark}
\label{remant}
By looking at the proof given below, one can see that it would be enough to require that $\partial S(a)$ exists, and that 
for {\it some} vector $\tilde \Phi \in \H$ inducing the same state as $\Phi$ on $\M(a)$ and inducing a state $\varphi'$ on $\M(a)'$, 
we have $\tilde \Phi \in \D(P) \cap \D(\log \Delta_{\tilde \Phi, \Omega; \M(a)'})$, and  $(D\omega' : D\tilde \varphi')_s \tilde \Phi \in \D(P)$.
\end{remark}
\begin{remark}
One can likewise show that if $\partial \bar S(a)$ exists, 
$\Phi \in \D(P) \cap \D(\log \Delta_a)$, and that $u_{s} \Phi \in \D(P)$ for 
$s \in (-\infty, s_0)$ for some $s_0$, then 
\ben
\label{antright}
\partial \bar S(a) = 2\pi \inf_{u \in \M(a)} (u\Phi, P u\Phi),
\een
where $u_s = (D\omega : D\varphi)_s$ is the Connes-cocycle for $\M(a)$, and $\bar S(a) = S(\Phi|\!| \Om)_{\M(a)'}$.
\end{remark}
\begin{remark}
\label{remarkdomain}
    The assumptions of our theorem \ref{antthm} are stronger than those by \cite{CF18},  
    who only ask that $(\Phi, P\Phi), S(a), \bar S(a)<\infty$. It is unclear to us whether the explicit formulas of proposition \ref{prop2}, or even just lemma \ref{lem3}, would apply under their assumptions. \end{remark}
\begin{proof}
We give a proof of \eqref{antright}. The other case is treated in a completely analogous manner, 
suitably exchanging $\M(a)$ with $\M(a)'$ and $S$ with $\bar S$.  For a function $f: \RR \to \RR$, we define the left and right derivatives 
as $\partial^\pm f(a) = \lim_{h\to 0\pm} [f(a+h)-f(a)]/h$ where these limits exist. Let $\bar S_u(a)$ be the relative entropy with $\Phi$ 
replaced by $u\Phi$. Then 
\ben
\label{pus}
\partial \bar S(a) = \partial^- \bar S(a) = \partial^- \bar S_u(a)
\een
by the invariance \eqref{Sdefcov} of the relative entropy, since $u$ is an isometry from $\M(a)$. 
The sum rule \eqref{dS1} and the monotonicity 
of the relative entropy implies  
\ben
\label{dS3}
2\pi (u\Phi, Pu\Phi) = \partial^- \bar S_u(a) - \partial^- S_u(a) \ge \partial^- \bar S_u(a) = \partial \bar S(a).
\een
Therefore, 
\ben
\label{udom}
2\pi \inf_{u \in \M(a)} (u\Phi, Pu\Phi) \ge \partial \bar S(a).
\een
Thus, we must display a sequence $u_n \in \M(a)'$ of isometries such that equality is 
attained in this bound as $n \to \infty$. Following \cite{CF18}, we shall show that $u_n:= u_{s_n}$ 
for a sequence $s_n \to -\infty$ does the job. 
We will see that this follows straightforwardly from \ref{lem3eq}. 

For ease of notation, we shall put $a=0$. This is no loss of generality since $(\M(a) \subset \M(a+1),\Omega)$ is a half-sided modular inclusion with the same $P$.
Equation \eqref{pus} can be applied to $u_s \Phi$, because $u_s \Phi \in \D(P)$ by assumption and because of the following lemma:

\blem
We have $u_s \Phi \in \D(\log \Delta_{u_s\Phi, \Omega})$
under the assumptions of the theorem \eqref{antthm}.
\elem
\begin{proof}
By covariance \eqref{covariance} of the relative modular operator, 
\ben
\frac{1}{t} (\Delta_{u_s\Phi, \Omega}^{it}-1) u_s \Phi = u_s \frac{1}{t} (\Delta_{\Phi, \Omega}^{it}-1) \Phi, 
\een
so the limit $t\to 0$ exists in the strong sense since we assume that $\Phi \in \D(\log \Delta_{\Phi, \Omega})$. 
Consequently,  $u_s \Phi \in \D(\log \Delta_{u_s\Phi, \Omega}) $ by Stone's theorem.
\end{proof}
Applying \eqref{pus} and then  \eqref{lem3eq} to $u_s \Phi$, we get the second inequality in 
\ben
\label{udom1}
0\le 2\pi (u_s\Phi, Pu_s\Phi) - \partial \bar S(\Phi | \! | \Omega) \le \left(u_s \Phi,\{2\pi P - i\left[P, \log \Delta_{u_s\Phi, \Omega} \right] \} u_s \Phi \right).
\een
The first inequality is \eqref{udom} applied to $u=u_s$.
\blem
\label{lem4}
Under the assumptions of the theorem \eqref{antthm}, we have 
\ben
\label{udom2}
\left(u_s \Phi,\{2\pi P - i\left[P, \log \Delta_{u_s\Phi, \Omega} \right] \} u_s \Phi \right)=
e^{2\pi s}\left(\Phi,\{2\pi P - i\left[P, \log \Delta_{\Phi, \Omega} \right] \} \Phi \right).
\een
\elem 
The proof of this lemma is given below. 

\noindent
{\bf End of proof of theorem \ref{antthm}.}
Taking $s \to -\infty$ in \eqref{udom1} and using \eqref{udom2} gives
\ben
0 = \lim_{s \to -\infty} 2\pi (u_s\Phi, Pu_s\Phi) - \partial \bar S(\Phi | \! | \Omega).
\een
This equation shows that the infimum in \eqref{udom} is achieved by the sequence $u_n = u_{s_n}$ provided that $s_n \to -\infty$, and that it precisely saturates 
 \eqref{udom}.
\end{proof}

\begin{remark}
We remark that combining \eqref{prop2s2} and lemma \ref{lem4}, we have also obtained the interesting balance-type formula \cite[Eq. 35]{CF18}, under the hypothesis on $\Phi$ of proposition \ref{prop2} as required for \eqref{prop2s2}, although in \cite{CF18}, this balance-type formula is obtained under weaker hypothesis on $\Phi$. The hypothesis in theorem \ref{antthm} are also weaker than those of proposition \ref{prop2}, which is why we can only apply the inequality \eqref{lem3eq} above in \eqref{udom1}, rather than the equality \eqref{prop2s2}, which however is still sufficient for the proof of theorem \ref{antthm}. 
\end{remark}

\noindent
{\bf Proof of Lemma \ref{lem4}.}
At first, for simplicity, we assume that $\Phi$ is cyclic and separating for $\M$.
We study the family of bounded operators depending on $t,a \in \RR$, defined by 
\ben
\label{gta}
g(t,a) = U(a) \Delta^{-it}_{\Phi,\Omega} U(-ae^{-2\pi t}) \Delta^{it}_{\Phi,\Omega},  
\een
and considered also by \cite{CF18} and by \cite{Florig} for $\Phi=\Omega$. Each $g(t,a)$ is a unitary which, as we shall now show, commutes with any $m' \in \M'$ as long as $a \ge 0$, hence
$g(t,a) \in \M$ by the double commutant theorem. To see this, we note that $\Delta^{it}_{\Phi,\Omega}m'\Delta^{-it}_{\Phi,\Omega}
= \Delta^{it}_{\Omega}m'\Delta^{-it}_{\Omega}$ by \eqref{mflow1}, \eqref{mflow2}. Furthermore, for any $b \ge 0$, we have $U(-b) \M' U(b) \subset \M'$ by the 
properties of half-sided modular inclusions. The statement $[m',g(t,a)]=0$ for $a \ge 0$ then follows from the commutation relations between $U(b), \Delta_\Omega^{it}$ for 
half-sided modular inclusions.

Now, let $g_s(t,a)$ be defined as $g(t,a)$, but replacing $\Phi$ by $u_s \Phi$, so $g_s(t,a)$ is still a unitary of $\M$ for $a\ge 0$. Using the first of the equivalent forms
\ben
\label{usform2}
u_s = \Delta_{\Omega,\Phi}^{is} \Delta_\Phi^{-is} = \Delta_{\Omega}^{is} \Delta_{\Phi,\Omega}^{-is},
\een
together with $\Delta_\Phi^{-is}\Phi = \Phi,$ and  
\ben
\Delta_{\Omega,\Phi}^{-is}g_s(t,a)\Delta_{\Omega,\Phi}^{is} = \Delta_{\Omega}^{-is}g_s(t,a)\Delta_{\Omega}^{is}
\een
since $g_s(t,a) \in \M$, we arrive at
\ben
\label{usform3}
(u_s \Phi, g_s(t,a) u_s\Phi) = (\Phi, \Delta_{\Omega}^{-is}g_s(t,a)\Delta_{\Omega}^{is} \Phi).
\een
We next use $\Delta_{u_s\Phi, \Omega}^{it} = u_s \Delta_{\Phi, \Omega}^{it} u_s^*$ from covariance and then, using the second form in \eqref{usform2},
\ben
\Delta_{u_s\Phi, \Omega}^{it} = \Delta_{\Omega}^{is} \Delta_{\Phi, \Omega}^{it}\Delta_{\Omega}^{-is}.
\een
Applying this relation to $g_s(t,a)$ and using the commutation relations for half-sided modular inclusions gives
\ben
\begin{split}
\Delta_{\Omega}^{-is}g_s(t,a)\Delta_{\Omega}^{is} &=  \Delta_{\Omega}^{-is} U(a) \Delta_{\Omega}^{is} \Delta^{-it}_{\Phi,\Omega} \Delta_{\Omega}^{-is} U(-ae^{-2\pi t}) \Delta_{\Omega}^{is} \Delta^{it}_{\Phi,\Omega}\\
&= U(ae^{2\pi s}) \Delta^{-it}_{\Phi,\Omega} U(-ae^{-2\pi(t-s)}) \Delta^{it}_{\Phi,\Omega}.
\end{split}
\een
Inserting this relation into \eqref{usform3} gives
\ben
\label{usform5}
(u_s \Phi, U(a) \Delta^{-it}_{u_s\Phi,\Omega} U(-ae^{-2\pi t}) \Delta^{it}_{u_s\Phi,\Omega} u_s\Phi) = (\Phi, U(ae^{2\pi s}) \Delta^{-it}_{\Phi,\Omega} U(-ae^{-2\pi(t-s)}) \Delta^{it}_{\Phi,\Omega} \Phi).
\een
Now, the relation
\ben
Pu_s \Delta_{\Phi,\Omega}^{it}\Phi=e^{2\pi t} \Delta^{it}_\Omega P u_{s-t}\Phi,
\een
which is obtained using the second form in \eqref{usform2} and the commutation relations for half-sided modular inclusions, shows that $\Delta^{it}_{u_s\Phi,\Omega} u_s\Phi= u_s \Delta_{\Phi,\Omega}^{it}\Phi$
is in the domain of $P$ for all real $s \in (-\infty, s_0)$ and sufficiently small $|t|$ as $u_s \Phi$ is for all $s \in (-\infty, s_0)$. Consequently, by Stone's theorem, we may take the right derivative of both sides of 
\eqref{usform5} with respect to $a$ at $a=0$, to obtain
\ben
\label{usform4}
(u_s \Phi, \{ iP - ie^{-2\pi t} \Delta^{-it}_{u_s\Phi,\Omega} P  \Delta^{it}_{u_s\Phi,\Omega}\} u_s\Phi) = 
e^{2\pi s}  (\Phi, \{ iP  - ie^{-2\pi t} \Delta^{-it}_{\Phi,\Omega} P  \Delta^{it}_{\Phi,\Omega}\} \Phi).
\een
Finally, as $u_s \Phi$ is in the domain of $\log \Delta_{u_s\Phi,\Omega}$ for all $s \in (-\infty, s_0)$, we may take the derivative with respect to $t$ at $t=0$ of both sides of the equation by Stone's theorem. This gives the statement of the lemma.

If $\Phi$ is not cyclic or not separating, then straightforward modifications of the above formulas 
involving the so-called support projections $s(\Phi) \in \M, s'(\Phi) \in \M'$ must be made when dealing 
with relative modular operators, conjugations, and flows. 
For details on such relations see e.g., \cite[App. A]{CF18}. These modifications do not change the argument significantly and are therefore not laid out in detail here.
\eproof

Our proof of theorem \ref{antthm} yields the following corollary:
\bcor
Under the assumptions of theorem \ref{antthm}, 
the infimum in \eqref{antright} is attained by the 
sequence of Connes-cycles $u_{s_n}$  
 and the infimum in \eqref{antleft} is attained by the sequence 
of Connes-cycles $u_{-s_n}'$, where $s_n \to -\infty$.
\ecor
\begin{remark}
    The flowed states $u_{s}\Phi$ respectively 
    $u_{-s}'\Phi$ 
    have an interesting holographic interpretation \cite{B3}.
\end{remark}

\section{Proof of the QNEC}
\subsection{Proof under a regularity condition}
We now prove the QNEC under a regularity condition.

\begin{theorem}
\label{thm:qnec}
The function $[b,\infty) \owns a \mapsto S(a) \in [0,\infty]$ is convex for any vector $\Phi$ and $b$ such that 
$\Phi = m'\Omega$ for some $m' \in \M(b)'$. 
\end{theorem}
\begin{proof}
The proof is based on the ant-formula, theorem \ref{antthm}. We will first construct a 
class of regular states to which the theorem  \ref{antthm} can be applied, and then we  remove the regulators one by one. 
Throughout the proof, we let $\varphi = (\Phi, \cdot \Phi), \omega = (\Omega, \cdot \Omega)$ be the states on $\M(b)$ associated with 
the vectors $\Phi, \Omega$. 

We first assume that $\varphi$ is $c$-comparable to $\omega$ on $\M(b)$, see \eqref{ccomp}. 
Then we regulate $\varphi$ using the unital, normal, completely positive linear map $T_\lambda: \M(b) \to \M(b)$ 
defined by
\ben
\label{Tladef}
T_\lambda(m) := \lambda \int_0^\infty e^{-\lambda a} U(a)mU(a)^* \, da,
\een
where $\lambda >0$. The properties of this regulator are analyzed in the appendix. We find, in particular, that $\lim_{\lambda \to \infty} T_\lambda(m)=m$
in the strong topology for any $m \in \M$. A corresponding regulated state on $\M(b)$ is then defined as 
\ben
\varphi_\lambda(m) = \varphi(T_\lambda(m)), 
\een
so that $\lim_{\lambda \to \infty} \varphi_\lambda = \varphi$ in the weak topology on $\M^+_*$. Since by construction we have that $\omega(T_\lambda(m)) = \omega(m)$, 
it follows that $\varphi_\lambda$ is $c$-comparable to $\omega$ on $\M(b)$. 

By construction, $T_\lambda$ restricts to a unital, normal, completely positive linear map $T_\lambda: \M(a) \to \M(a)$
for any $a \ge b$, and, considered as a state on $\M(a)$, $\varphi_\lambda$ is still $c$-comparable to $\omega$. 
Let $\Phi_{\lambda}(a)$ be the vector  representative of $\varphi_\lambda$ in the natural cone $\P^\sharp_\Omega[\M(a)]$, i.e., 
when considered as a state on $\M(a)$ for any $a \ge b$. By applying theorem \ref{thmA1} to $\M(a)$ instead of $\M$, we have 
that $\Phi_{\lambda}(a) \in \D(P) \cap \D(\log \Delta'_a)$, 
and as a consequence of lemma \ref{lem3.1}, it follows that 
\ben
\label{regS}
S_\lambda(a) := S(\varphi_\lambda |\!| \omega)_{\M(a)}
\een
is continuously differentiable on $[b,\infty)$.
Furthermore, by theorem \ref{thmA1},
 if we let $u^{\prime}_{\lambda,a,s}$ be the Connes-cocycle \eqref{uDD} associated with $\M(a)'$, $\Phi_{\lambda}(a), \Omega$, 
then $u^{\prime}_{\lambda, a, s} \Phi_\lambda(a) \in \D(P)$, for any $a \ge b$. 
By theorem \ref{antthm} and remark \ref{remant}, we therefore have the ant-formula 
\ben
\label{antleft1}
-\partial S_\lambda(a) = 2\pi \inf_{u' \in \M(a)'} (u'\Phi_\lambda, P u'\Phi_\lambda),
\een
for {\it any} vector representative $\Phi_\lambda$ of $\varphi_\lambda$ on $\M(b)$ and any $a \ge b$. 

Since $\M(a_1)' \subset \M(a_2)'$ for $a_1 \le a_2$, \eqref{antleft1} implies that $\partial S_\lambda(a_1) \le \partial S_\lambda(a_2)$, 
so $S_\lambda$ is convex on $[b,\infty)$. Now, in view of the monotonicity of the relative entropy, one has 
\ben
S_\lambda(a) = S(\varphi \circ T_\lambda |\!| \omega)_{\M(a)} = S(\varphi \circ T_\lambda |\!| \omega \circ T_\lambda )_{\M(a)} \le S(\varphi  |\!| \omega)_{\M(a)} = S(a). 
\een
Then, by the lower semi-continuity of the relative entropy, one has
\ben
S(a) \le \liminf_{\lambda \to \infty} S_\lambda(a) \le \limsup_{\lambda \to \infty} S_\lambda(a) \le S(a), 
\een
so $\lim_{\lambda \to \infty} S_\lambda(a) = S(a)$ for any $a \ge b$. Since the  limit of a pointwise convergent sequence of convex function is convex, it follows that $S: [b,\infty) \to 
[0,\infty]$ is convex for any $\varphi$ that is $c$-comparable to $\omega$ for some $c>0$. 

Now consider vectors of the form $\Phi = m'\Omega$ where $m' \in \M(b)'$. To reduce this situation to the previous case, we employ the well-known trick to set $\varphi_{\epsilon} = \epsilon \omega + (1-\epsilon) \varphi$ for a small $\epsilon>0$, which is $c$-comparable to $\omega$ on $\M(b)$ and hence on any $\M(a)$ for $a \ge b$ for a suitable $c$. We define $S_{\epsilon}$ by \eqref{regS} for the state $\varphi_{\epsilon}$, so by the previous argument, $S_\epsilon$ is convex on $[b,\infty)$.  By combining the lower semi-continuity and convexity of the relative entropy functional $\M(a)_* \owns \varphi \mapsto S(\varphi |\!| \omega)_{\M(a)} \in [0,\infty]$, we find $\lim_{\epsilon \to 0+} S_{\epsilon}(a) = S(a)$ \cite{Ar77}. Since the limit of a pointwise convergent sequence of convex functions is convex, $S_{}: [b,\infty) \to [0,\infty]$ is convex. 
\end{proof}
\subsection{Weakening the regularity condition in Thm. \ref{thm:qnec}}

In this subsection, we show that theorem \ref{thm:qnec} remains true for states $\Phi$ in a larger domain. Our domain is mathematically natural and related to the 
domain of the square root of the relative modular operator, but we note that our restrictions are still stronger than those 
required by \cite{CF18} in their proof of the QNEC. We first give a characterization of our domain. 

Let $\M$ be a von Neumann algebra on the Hilbert space $\H$ and $\Om,\Phi\in\H$ cyclic and separating vectors for $\M$. Recall the anti-linear operators
$S_{\Om,\Phi}  : m\Phi \mapsto m^*\Om\, , \quad m\in \M\, ,
S'_{\Om,\Phi} : m'\Phi \mapsto m'^*\Om\, , \quad m'\in \M'\,$
and use the same symbol for their closures, see section \ref{modthbas}. Then we have polar decompositions
as in  \eqref{poldec}, and we recall that
\ben\label{JJ3}
S^*_{\Om,\Phi} = S'_{\Om,\Phi}\, ,\qquad {\Delta}'_{\Om,\Phi} = {\Delta}^{- 1}_{\Phi,\Om}\, ;
\een
so
\ben\label{pdS}
\D(S_{\Om,\Phi}) = \D(\Delta^{1/2}_{\Om,\Phi}) = \D({\Delta}^{-1/2}_{\Phi,\Om})\, .
\een
We write $T\, \hat\in \, \M$ if  $T$ is an operator affiliated with $\M$, that is, $T$ is a closed, densely defined operator such that
 $u'Tu'^* = T$ for all unitaries $u'\in \M'$;
equivalently, $Tm' \supset m'T$ for all $m'\in\M'$. 

We define 
\[
\C(\M;\Phi,\Om) := \big\{T\,\hat\in\, \M: \Om\in D(T), \Phi\in D(T^*)\big\}\,;
\]
clearly $\C(\M;\Phi,\Om) \supset \M$. 

We now show that every $\Psi\in \D(\Delta^{1/2}_{\Phi,\Om})$ has the form $\Psi = T\Om$ for some $T\,\hat\in\, \M$. 
\begin{proposition}
\label{propdom}
$\D(\Delta^{1/2}_{\Phi,\Om}) = \D(\Delta'^{-1/2}_{\Om,\Phi}) = \C(\M;\Phi,\Om)\Om$, and we have $S_{\Phi,\Om}T\Om = T^*\Phi$ for every $T\in \C(\M;\Phi,\Om)$. 
\end{proposition}
\proof
We show that $\C(\M;\Phi,\Om)\Om = \D(S_{\Phi,\Om})$, the proposition then follows by \eqref{pdS}. 

$\bullet$ $\C(\M;\Phi,\Om)\Om \subset \D(S_{\Phi,\Om})$:  With $T\in \C(\M;\Phi,\Om)$, by the first identity in \eqref{JJ3} it suffices to show that $T\Om\in D(S'^*_{\Phi,\Om})$.
This follows because for all $m'\in\M'$ we have the identities
\[
(T\Om, S'_{\Phi,\Om} m'\Om) = (T \Om, m'^* \Phi) = ( m' T \Om, \Phi) = (T m' \Om, \Phi)  = (m '\Om, T^*\Phi) 
= (m '\Om, S_{\Phi,\Om} T\Om)
\]
and the fact  that $\M'\Om$ is a core for $S'_{\Phi,\Om}$. This also shows that $S_{\Phi,\Om}T\Om = T^*\Phi$. 

$\bullet$ $\D(S_{\Phi,\Om}) \subset \C(\M;\Phi,\Om)\Om$:  Let $\Psi\in \D(S_{\Phi,\Om})$ and define the linear maps $T_1$, $T_2$
\ben\label{T1T2}
T_1: m'\Om \mapsto m'\Psi\, ,\qquad T_2 : m'\Phi \mapsto m' S_{\Phi,\Om}\Psi\, ,\qquad m' \in \M\, ,
\een
with domains $\M'\Om$ and $\M'\Psi$,
which are well and densely defined since both $\Om$ and $\Phi$ are cyclic and separating for $\M'$. We have
\begin{multline}
(m'_2\Phi, T_1 m'_1\Om) = (m'_2\Phi, m'_1\Psi) =  (m'^*_1 m'_2\Phi, \Psi)   =  (S'_{\Phi,\Om}m'^*_2 m_1\Om, \Psi) \\
=  (S^*_{\Phi,\Om}m'^*_2 m_1\Om, \Psi) = (S_{\Phi,\Om}\Psi, m'^*_2 m'_1\Om)  = (m'_2 S_{\Phi,\Om}\Psi,  m'_1\Om) 
 = (T_2 m'_2 \Phi,  m'_1\Om) 
\end{multline}
with $m_1 , m_2\in\M$, thus $T_1 \subset T^*_2$, $T_2 \subset T^*_1$, and $T_1$, $T_2$ are closable. Let $T$ be the closure of $T_1$; since
\[
u' T_1 u'^* m' \Om = u' u'^* m' \Psi =  m' \Psi  = T_1 m'\Om\, , \quad m'\in \M' \,  ,
\]
for all unitaries $u'\in\M'$, we have $u' T_1 u'^* = T_1$, thus $u' T u'^* = T$, that is $T\, \hat\in\, \M$. 

From \eqref{T1T2} we see that $T\Om = \Psi$ and $S_{\Phi,\Om}\Psi = T_2\Phi = T^*\Phi$, thus $T\in  \C(\M;\Phi,\Om)$. 
\eproof

\begin{theorem}
Theorem \ref{thm:qnec} remains true for states of the form 
$\Phi = T'\Omega$ when $T' \hat\in \M(b)'$ and $\Omega \in \D(T')$. 
The vectors $\Phi$ such that $\Phi \in \D(\Delta^{\prime 1/2}_{b})$, 
where $\Delta^{\prime}_{b} = \Delta_{\Phi,\Omega;\M(b)'}$, all have this form.
\end{theorem}
\begin{remark}
Note that $\Phi \in \D(\Delta^{\prime 1/2}_{b})$ implies that $S(b)<\infty$, hence $S(a)<\infty$ for all $a \ge b$ by monotonicity of 
the relative entropy.
\end{remark}

\proof The second statement follows immediately from proposition \ref{propdom}. 

For the first statement, we construct a sequence of regulated states which enables us to then make an approximation argument suggested to us by \cite{Wirth2}. 
We first make a polar decomposition $T' = v'|T'|$ and let $e_n'$ be the spectral projection of $|T'|$ associated with $[0,n]$, 
where $n \in \mathbb{N}$. Then $e_n', |T'|e_n'$ and $v'$ are in $\M(b)'$ by well-known characterizations of operators affiliated with a von Neumann algebra.
For $\epsilon>0$, set $\Phi_{n,\epsilon} = (\epsilon 1+|T'|^2e_n')^{1/2}\Omega$. The corresponding functionals $\varphi_{n,\epsilon}$ on 
$\M(b)$ satisfy (i) the QNEC by the proof of theorem \ref{thm:qnec},  (ii) $\varphi_{n,\epsilon}-\varphi_{k,\epsilon} \ge 0$ for $n \ge k$, (iii) $\lim_{n\to \infty} \varphi_{n,\epsilon} =\varphi_\epsilon$ in the weak sense, where $\varphi_\epsilon=\varphi + \epsilon \omega$, (iv) $\varphi_{n,\epsilon} \le \varphi_\epsilon$.
By \cite[Lem. 12.2 and Prop. 12.9]{OP} 
\ben
\label{weq}
S(\Phi_{n,\epsilon} |\! | \Omega)_{\M(b)} = S(\Phi_{n,\epsilon} |\! | \Phi_\epsilon)_{\M(b)} + \varphi_{n,\epsilon}(h_{\epsilon}).
\een
Here $\Phi_\epsilon$ is a vector representing $\varphi_\epsilon$, and $h_{\epsilon} = \log \Delta_{\Phi_\epsilon,\Omega;\M(b)}- \log \Delta_{\Omega;\M(b)}$, which is an extended-valued lower bounded self-adjoint operator affiliated with $\M(b)$, see, e.g., \cite[Ch. 12]{OP} (lower bound $(\log \epsilon)1$). By (iii),(iv), and lower semi-continuity of the relative entropy, 
\ben
0 = S(\Phi_{\epsilon} |\! | \Phi_\epsilon)_{\M(b)} \le \liminf_{n \to \infty} S(\Phi_{n,\epsilon} |\! | \Phi_\epsilon)_{\M(b)} \le \limsup_{n \to \infty} S(\Phi_{n,\epsilon} |\! | \Phi_\epsilon)_{\M(b)} \le 0, 
\een
so $\lim_{n \to \infty} S(\Phi_{n,\epsilon} |\! | \Phi_\epsilon)_{\M(b)}=0$.
By (ii),(iii), and $h_\epsilon \ge (\log \epsilon)1$, we have
$\lim_{n \to \infty} \varphi_{n,\epsilon}(h_{\epsilon}) = \varphi_\epsilon(h_\epsilon)$. Thus, $\lim_{n \to \infty} S(\Phi_{n,\epsilon} |\! | \Omega)_{\M(b)} = S(\Phi_{\epsilon} |\! | \Omega)_{\M(b)}$, and then by the same argument as in the proof of theorem \ref{thm:qnec}, we get that
$\lim_{\epsilon \to 0}\lim_{n \to \infty} S(\Phi_{n,\epsilon} |\! | \Omega)_{\M(b)} = S(\Phi|\! | \Omega)_{\M(b)}$. In these arguments we may replace $b$ by any $a \ge b$ because (i)--(iv) still hold. Thus, by (i), $[b, \infty) \owns a \mapsto S(\Phi |\!| \Omega)_{\M(a)}$ is convex. 
\eproof

\section{Outlook}
\label{outlook}
Our proof of the QNEC was mainly based on the expressions for $\partial S(a)$ and $\partial \bar S(a)$, the relative-entropy-currents between $\Phi$
and $\Omega$ with respect to $\M(a)$ respectively $\M(a)'$, given in proposition \ref{prop2}. These expressions have a close similarity to 
Spohn's classical formula \cite{S77} for the entropy production under a Markov semi-group when applied to the setting of half-sided modular inclusions. We note that 
Spohn's formula has recently been analyzed for continuous Markov semi-groups on general sigma-finite von Neumann algebras by \cite{W25} in the context of 
logarithmic Sobolev inequalities. It would be interesting to consider possible connections of logarithmic Sobolev inequalities to the QNEC. 

We may write the expressions in proposition \ref{prop2} as expectation values, e.g., 
$\partial S(a) = (\Phi, \Sigma_\Phi(a) \Phi)$, where
\ben
\label{prop2s11}
\Sigma_{\Phi}(a) = i[P, \log \Delta_{a}' - \log \Delta_{\Phi,a}'] .
\een
This operator is defined as a quadratic form on the domain $\D$ in proposition \ref{prop2}. As such, it has a vanishing commutator with any sufficiently smooth element $m' \in \M(a)'$. So in this sense, $\Sigma_\Phi(a)$ is affiliated with $\M(a)$, although due to its unbounded nature, it is not an element of $\M(a)$ in general. 

In view of $\partial S(a) = (\Phi, \Sigma_\Phi(a) \Phi)$, we may perhaps think of $\Sigma_\Phi$ as a relative-entropy-current operator. Adopting this interpretation, one may think  of 
\ben
{\rm Var}[\partial S] = (\Phi, \Sigma_\Phi^2 \Phi) - (\Phi, \Sigma_\Phi \Phi)^2
\een
as the variance of the relative-entropy-current, or simply the variance of the QNEC. Of course, since $\Sigma_\Phi(a)$ is only a quadratic form, this quantity may be infinite even if $\Phi \in \D$, although it is probably finite for sufficiently smooth vectors. 

Leaving this technical issue aside, our interpretation is supported by the fact that, for sufficiently smooth isometries $u' \in \M(a)'$, one has $\Sigma_{u'\Phi}(a) = u'\Sigma_\Phi(a)u'^*$, so this variance only depends on $\varphi = (\Phi, \ . \ \Phi)$ viewed as a state on $\M(a)$, 
and not on the particular vector representative.
We think that it would be worth further investigating this variance in the context of quantum gravity.

In view of frequent applications in the context of the replica method, it would also be worth investigating the validity of the QNEC for the 
R\'enyi entropies of index $n \in \mathbb{N}$, see \cite{R} for a heuristic proof in the case of free QFTs. 
\medskip

\noindent
{\bf Acknowledgements.} 
R.L. acknowledges the MIUR Excellence Department Project awarded to the Department of Mathematics, University of Rome Tor Vergata, CUP E83C23000330006.
S.H. warmly thanks Edward Witten for exchanges regarding the QNEC.

\noindent
{\bf Data availability.} Data sharing not applicable to this article as no data sets were generated or analysed during the current study.

\medskip

\noindent
{\bf Declarations}

\medskip

\noindent
{\bf Conflict of interest.} The authors have no relevant financial or non-financial interests to disclose.

\appendix

\section{Properties of $T_\lambda$}
\label{app:A}

In this appendix, we assume that $(\N \subset \M, \Omega)$ is a half-sided modular inclusion. Then we get $U(a)=e^{iaP}$ and 
we can define $T_\lambda: \M \to \M$ by \eqref{Tladef} for any $\lambda > 0$. We already 
noted that $T_\lambda$ is a unital, normal, completely positive map for real 
$\lambda > 0$.  Following \cite{P85}, \cite{P88} one may define 
\ben
\label{Vdef}
V m\Omega := T(m)\Omega, \quad m \in \M, 
\een
for any unital, normal, completely positive map $T$ on $\M$, and 
the Schwarz inequality $T(m)T(m^*) \ge T(mm^*)$ can be used to show that 
$V$ extends to a bounded linear operator on $\H$ with norm 
$\| V \| \le 1$. It is well-known \cite{P85}, \cite{P88} that $V$ plays nicely with the relative modular operators
involving the state $\Omega$ via the theory of operator monotone functions, and we will use such 
results below.  

We apply \eqref{Vdef} to $T_\lambda$ and call the corresponding linear operator $V_\lambda$.
This operator also has a more explicit expression, which follows immediately by applying \eqref{Tladef}, \eqref{Vdef} to the dense subspace of vectors of the form 
$\xi = m\Omega, m \in \M$ and using $U(a)\Omega = \Omega$:
\ben
\label{Vform}
V_\lambda = \frac{\lambda}{\lambda -iP}.
\een
From $J_\Omega P J_\Omega = P$, we then also get 
\ben
J_\Omega V_\lambda J_\Omega = V_\lambda^* = \frac{\lambda}{\lambda +iP}.
\een
Given a state $\varphi$ on $\M$, we set $\varphi_\lambda(m):= \varphi(T_\lambda(m))$, and we let $\Phi_\lambda \in \P^\sharp_\Omega$ be 
its representer in the natural cone. It is obvious that $\omega_\lambda = \omega$, 
where $\omega = (\Omega, \cdot \Omega)$ is the state on $\M$ induced by $\Omega$.
We also set $\varphi'_\lambda(m') = (\Phi_\lambda, m' \Phi_\lambda) = \varphi_\lambda(J_\Omega m'J_\Omega), m' \in \M'$, 
which defines a state on $\M'$. The Connes-cocycles \eqref{uDD} associated with $\varphi'_\lambda$ and $\varphi_\lambda$ are 
denoted by $u'_{\lambda, s}$ and $u_{\lambda, s}$.

In the following we assume that $\varphi$ is $c$-comparable to $\omega$ (see \eqref{ccomp}) for some fixed $c>0$.
\begin{theorem}
\label{thmA1}
$T_\lambda$ has the following properties for $\lambda>0$:
\begin{enumerate}
\item $\lim_{\lambda \to \infty} \| (T_\lambda(m) - m) \xi \| = 0$ for all $\xi \in \H, m \in \M$. 
\item $\varphi_\lambda$ is $c$-comparable to $\omega$ (see \eqref{ccomp}). 
\item If $h_\lambda = \log \Delta_{\Phi_\lambda, \Omega} - \log \Delta_\Omega$, 
then both $h_\lambda$ and $[P,h_\lambda]$ are  in $\M$.
\item $\Phi_\lambda \in \D(P)$.
\item $u'_{\lambda,s} \Phi_\lambda, u_{\lambda,s} \Phi_\lambda \in \D(P)$.
\item $\Phi_\lambda \in \D(\log \Delta_{\Phi_\lambda, \Omega}')$.
\end{enumerate}
\end{theorem}
\begin{remark}
Using connections with ergodic theorems for von Neumann algebras, one can also 
obtain $\lim_{\lambda \to 0} T_\lambda(m) = \omega(m)1$ weakly,  provided $\M\cap\N' = \CC$, but we will not need such 
a relation in this work.
\end{remark}
\begin{proof}
1) follows from the strong continuity of $U(a)$ and the definition \eqref{Tladef}.

2) follows from the fact that $\omega_\lambda = \omega$.

3) As is well-known, the fact that $\varphi_\lambda$ is $c$-comparable to $\omega$ implies $h_\lambda \in \M$ and in fact $\| h_\lambda \| \le \log c^{-1}<\infty$ \cite{Araki4}, see e.g., \cite[Sec. 12]{OP} for a summary of related material.
To show that $[P,h_\lambda]$ is bounded, we proceed in several steps. First we pass from $\Phi$ to $c\Phi$, whereby  $h_\lambda$
changes to $h_\lambda + (\log c)1 \le 0$. It is obviously enough to show that $[P,h_\lambda]$ is bounded for this rescaled state. 

It is known by general constructions associated with operators $V$ of the form \eqref{Vdef} that 
\ben
(V_\lambda \xi, \Delta_{\Phi, \Omega} V_\lambda \xi) \le (\xi, \Delta_{\Phi_\lambda, \Omega}  \xi), \quad \xi \in \D(\Delta^{1/2}_{\Phi_\lambda, \Omega}), 
\een
and that this implies \cite[Proof of Lem. 1.2]{OP}
\ben
\label{80}
-( \xi, V_\lambda^*(\mu + \Delta_{\Phi, \Omega})^{-1} V_\lambda \xi) - \mu^{-1}
(\xi, (1-V_\lambda^* V_{\lambda}^{}) \xi)
\le -(\xi, (\mu+\Delta_{\Phi_\lambda, \Omega})^{-1}  \xi)
\een
for every $\mu>0$ and all $\xi \in \H$. Now we apply the representation \eqref{L1} of the logarithm, which gives
\ben
\label{81}
( \xi, V_\lambda^*(\log \Delta_{\Phi, \Omega}) V_\lambda \xi) 
\le (\xi, \log \Delta_{\Phi_\lambda, \Omega}  \xi)
\een
for all\footnote{\eqref{L1} requires ${\rm ker} A=\{0\}$ which holds in our case because 
$\varphi_\lambda$ is comparable to $\omega$.} $\xi$ such that the expectation values of the logarithms are defined. 
We may also apply this to the special case when $\Phi=\Omega$, which gives 
\ben
(\xi, V_\lambda^* (\log \Delta_{\Omega}) V_\lambda \xi) \le (\xi, \log \Delta_{\Omega}  \xi).
\een
Replacing $\xi$ by $J_\Omega \xi$ in this last inequality, using that $J_\Omega (\log \Delta_\Omega) J_\Omega = -\log \Delta_\Omega$ and that $J_\Omega V_\lambda J_\Omega = V_\lambda^*$ results in 
\ben
\label{83}
-( \xi, V_\lambda (\log \Delta_{\Omega}) V_\lambda^* \xi) \le -(\xi, \log \Delta_{\Omega}  \xi). 
\een
Next, the relations for half-sided modular inclusions give $[\log \Delta_\Omega, iP]= -2\pi P$, and thereby
\ben
\label{chain}
\begin{split}
&V_\lambda (\log \Delta_{\Omega}) V_\lambda^* \\
= &\, V_\lambda V_\lambda^* \log \Delta_{\Omega} + V_\lambda [\log \Delta_{\Omega}, V_\lambda^*]\\
= &\, V_\lambda^* V_\lambda \log \Delta_{\Omega} + V_\lambda [\log \Delta_{\Omega}, V_\lambda^*]\\
= & \,  V_\lambda^* (\log \Delta_{\Omega}) V_\lambda + V_\lambda [\log \Delta_{\Omega}, V_\lambda^*]-V_\lambda^* [\log \Delta_{\Omega}, V_\lambda]\\
= & \, V_\lambda^* (\log \Delta_{\Omega}) V_\lambda + 4\pi \lambda^3P (\lambda^2 + P^2)^{-2}\\
\le & \, V_\lambda^* (\log \Delta_{\Omega}) V_\lambda + 4\pi V_\lambda^* V_\lambda
\end{split}
\een
Using this identity when adding \eqref{83} and \eqref{81} gives 
\ben
\label{84}
-4\pi (V_\lambda \xi, V_\lambda \xi) + (V_\lambda \xi, (\log \Delta_{\Phi, \Omega} - \log \Delta_{\Omega}) V_\lambda \xi) \le (\xi, (\log \Delta_{\Phi_\lambda, \Omega} -\log \Delta_\Omega) \xi) = (\xi, h_\lambda \xi). 
\een
Using that $h_\lambda \le 0$ and that $\| V_\lambda (1\pm i\lambda^{-1}P)\| =1$ and $\| \log \Delta_{\Phi, \Omega} - \log \Delta_{\Omega}\| \le 2\log c^{-1}$,
 we see that $\| (1\pm i\lambda^{-1}P) h_\lambda (1\pm i\lambda^{-1}P)\| \le 4\pi + 2\log c^{-1}$, and therefore, that 
 \begin{multline}
 |( \xi, [P, h_\lambda]\xi)| = 
 \frac{\lambda}{2} \Big|
 \Big(V_\lambda (1+ i\lambda^{-1}P) \xi, h_\lambda V_\lambda (1+ i\lambda^{-1}P) \xi \Big)\\
 -
 \Big(V_\lambda (1- i\lambda^{-1}P) \xi, h_\lambda V_\lambda (1- i\lambda^{-1}P) \xi \Big)
 \Big|
 \le (4\pi + 2\log c^{-1})\lambda \| \xi \|^2.
 \end{multline}
 This shows $\| [P, h_\lambda] \| \le(4\pi + 2\log c^{-1})\lambda<\infty$. To show that $[P, h_\lambda]\in \M$ we use 
 $[iP, h_\lambda] =  \lim_{a \to 0} (U(a)h_\lambda U(a)^*-h_\lambda)/a$ in the strong sense 
 and the fact that $U(a)h_\lambda U(a)^* \in \M(a)$ from $h_\lambda \in \M$ and the properties of half-sided modular inclusions. 
 This concludes the proof of 3).
 
4) It is known \cite{Araki4, Araki5}, see also, e.g., \cite[Sec. 12]{OP}, that every state $\Psi \in \P^\sharp_\Omega$ in the natural cone such that $\psi = (\Psi, \cdot \Psi)$ is $c$-comparable to $\omega$ for some $c>0$
has the representation 
\ben \label{psih}
\Psi = \Omega^{h} := \sum_{n=0}^\infty \int_0^{1/2} ds_1 \int_0^{s_1} ds_2 \dots \int_0^{s_{n-1}} ds_n \, \Delta_\Omega^{s_n}h\Delta^{s_{n-1}-s_n}_\Omega h
\cdots \Delta^{s_{1}-s_2}_\Omega h \Omega,
\een
where $h = \Delta_{\Psi,\Omega}-\Delta_\Omega$ is a bounded operator (with norm at most $\log c^{-1}$). The series is in fact absolutely convergent in norm. 
This follows from the more general fact \cite[Lem. A]{AM82} that the function 
\ben
\label{Fdef}
F(z_1, \dots, z_n):= \left(
\Delta^{\bar z_j'}_\Omega x_j \Delta^{\bar z_{j+1}}_\Omega x_{j+1} \cdots  \Delta^{\bar z_{n}}_\Omega x_n \Omega, \ 
\Delta^{z_j''}_\Omega x_j \Delta^{z_{j-1}}_\Omega x_{j-1} \cdots \Delta^{z_{1}}_\Omega x_1 \Omega
\right)
\een
is well-defined, holomorphic, and independent of the subdivision $z_j=z_j'+z_j''$, in the domain defined by
\begin{subequations}
\label{domainz}
\begin{align}
&1/2>{\rm Re} z_1, \dots, {\rm Re} z_n > 0,\\
&{\rm Re} z_1 + \cdots +  {\rm Re} z_{j-1} +  {\rm Re} z_{j}' < 1/2, \\
&{\rm Re} z_n + \cdots +  {\rm Re} z_{j+1} +  {\rm Re} z_{j}'' < 1/2,
\end{align}
\end{subequations}
for any $x_1, \dots, x_n \in \M$. $F$ is continuous on the closure of this domain, where it
satisfies 
\ben
\label{Fbound}
|F(z_1, \dots, z_n)| \le \|x_1\| \cdots \| x_n \|.
\een
We now specialize the representation \eqref{psih} to our vector $\Psi = \Phi_\lambda$
and we abbreviate the corresponding series symbolically as
\ben
\label{uT}
\Phi_\lambda  = \bar {\rm T} \exp\left[\int_0^{1/2} h_\lambda (is) ds \right] \Omega, 
\een
where ``$\bar T$ exp'' denotes the ``(anti-) time-ordered exponential'' obtained by ordering the integration parameters as in \eqref{psih}, 
and $h_\lambda(s):=\sigma^\omega_s(h_\lambda) = \Delta^{is}_\Omega h_\lambda \Delta_{\Omega}^{-is}$. We may also define the more general vectors 
\ben
\label{uT1}
\Phi_\lambda(z)  = \bar {\rm T} \exp\left[i\int_0^{z} h_\lambda (w) dw \right] \Omega, 
\een
for $z$ in the strip $\mathbb{S}_{1/2} = \{z: {\rm Im} z \in [-1/2,0]\}$, where the integration is along the straight path $t \mapsto tz=w(t)$, and where the series defining 
the anti-path-ordered exponential is again seen to be absolutely convergent in norm using \eqref{Fdef}, \eqref{domainz}, \eqref{Fbound},
by making appropriate choices of $F$. Then $\Phi_\lambda(z)$ is a norm-continuous vector valued function in $z$
the strip ${\mathbb S}_{1/2}$ that is holomorphic in the interior, and $\Phi_\lambda = \Phi_\lambda(-i/2)$. 
Furthermore, using $\Delta^{-it}_\Omega P \Delta^{it}_\Omega = e^{2\pi t} P$ for $t \in \RR$ and $P\Omega=0$, one can see easily that 
Duhamel's formula holds ($\xi \in \H$),
\begin{multline}
\label{94}
(\xi, P\Phi_\lambda(t)) = i\int_0^{t} e^{2\pi s}
\Bigg( \xi, \,
\bar {\rm T} \exp\left[ i \int_0^{s} h_\lambda(s_1) ds_1 \right] \times \\
\times \sigma^\omega_s([P,h_\lambda])\, \bar {\rm T} \exp\left[ i\int_s^{t} h_\lambda(s_2) ds_2 \right] 
\, \Omega \Bigg) ds.
\end{multline}
In fact, the series defining the right side after expanding out the exponentials 
is absolutely convergent by \eqref{Fdef}, \eqref{domainz}, \eqref{Fbound} and because $h_\lambda, [P,h_\lambda]$ are in $\M$, hence
in particular bounded by 3). By the same arguments, the series has a holomorphic extension to the interior of ${\mathbb S}_{1/2}$ which is continuous on 
${\mathbb S}_{1/2}$, and bounded by $\le e^{2\pi | {\rm Re} z|} |z|   \| [P,h_\lambda] \|  \|\xi\| e^{|z| \|h_\lambda\|}$. 
The function $z \mapsto (P\xi, \Phi_\lambda(z))$ defined for an arbitrary but fixed $\xi \in \D(P)$
also has a holomorphic extension to the interior of ${\mathbb S}_{1/2}$ which is continuous on 
${\mathbb S}_{1/2}$, and it coincides with the series defining the right side of \eqref{94} for $z \in \RR$ by definition. Therefore, by the edge-of-the-wedge theorem, the two must be identical for any $z \in {\mathbb S}_{1/2}$. 
In particular, we have  
\ben
|(P\xi, \Phi_\lambda(z))| \le e^{2\pi |{\rm Re} z|} |z| \|[P,h_\lambda]\| \|\xi\| e^{|z| \|h_\lambda\|} \le e^{2\pi |{\rm Re} z|} |z| (4\pi + 2\log c^{-1})\lambda c^{-2|z|} \| \xi \|
\een
for any $\xi \in \D(P)$, and any $z \in {\mathbb S}_{1/2}$.
Using this for $z=-i/2$ shows that $|(P\xi, \Phi_\lambda)| \le {\rm const.} \|\xi\|$, so $\Phi_\lambda \in \D(P)$, in fact
\ben
\| P\Phi_\lambda\| \le (2\pi + \log c^{-1})\lambda c^{-1}, 
\een
recalling that $c$ is the constant in \eqref{ccomp}. 

5) The proof is completely analogous to that of 4) and is based on the well-known representation
\ben
\label{uT1}
u_{\lambda,s} = {\rm T} \exp\left[-i\int_0^s h_\lambda(t) dt \right], 
\een 
and $u_{\lambda,s}' = J_\Omega u_{\lambda,-s} J_\Omega$. 

6) Since $\varphi_\lambda$ is $c$-comparable to $\omega$, there exists 
$m_\lambda' \in \M'$ for which $\Phi_\lambda = m'_\lambda \Omega$, and therefore 
$\Phi_\lambda \in \D(\Delta^{\prime 1/2}_{\Phi_\lambda, \Omega})$. 
By definition $\Phi_\lambda \in \D(\Delta^{1/2}_{\Omega, \Phi_\lambda})= \D(\Delta^{\prime -1/2}_{\Phi_\lambda, \Omega})$. 
Applying the spectral theorem to $\log \Delta^{\prime}_{\Phi_\lambda, \Omega}$ these two facts give the claim.
\end{proof}

\end{document}